\documentclass[12pt,reqno]{amsart}
\usepackage{txfonts,fullpage}
\usepackage[dvipdfmx]{graphicx}

\newtheorem{thm}{Theorem}[section]
\newtheorem{prop}[thm]{Proposition}

\newtheorem{fact}[thm]{Fact}

\newtheorem{rem}[thm]{Remark}
\numberwithin{equation}{section}
\allowdisplaybreaks

\begin{document}

\title{A Lax formulation of a generalized $q$-Garnier system}
\date{}
\author{Takao Suzuki}
\address{Department of Mathematics, Kindai University, 3-4-1, Kowakae, Higashi-Osaka, Osaka 577-8502, Japan}
\email{suzuki@math.kindai.ac.jp}

\maketitle

\begin{abstract}
Recently, a birational representation of an extended affine Weyl group of type $A_{mn-1}^{(1)}\times A_{m-1}^{(1)}\times A_{m-1}^{(1)}$ was proposed with the aid of a cluster mutation.
In this article we formulate this representation in a framework of a system of $q$-difference equations with $mn\times mn$ matrices.
This formulation is called a Lax form and is used to derive a generalization of the $q$-Garnier system.


\end{abstract}

\section{Introduction}

The Garnier system (in $n$ variables) was proposed as a generalization of the sixth Painlev\'e equation in \cite{G}.
It is derived from the isomonodromy deformation of a Fuchsian system of second order with $n+3$ regular singular points.
A $q$-analogue of the Garnier system was proposed from a viewpoint of a connection problem of a system of linear $q$-difference equations in \cite{Sak}.
Afterward it was studied in detail by a Pad\'e method in \cite{NY1,NY2}.

A group of symmetries for the Garnier system is isomorphic to the affine Weyl group $W(B_{n+3}^{(1)})$.
This fact was shown in \cite{K,Suz1}.
For the $q$-Garnier system, a symmetry structure was clarified recently.
In \cite{OS1} we formulated a birational representation of an extended affine Weyl group $\widetilde{W}(A_{2n+1}^{(1)}\times A_1^{(1)}\times A_1^{(1)})$ with the aid of a cluster mutation and derived the $q$-Garnier system as translations.
We also generalized this representation to that of $\widetilde{W}(A_{mn-1}^{(1)}\times A_{m-1}^{(1)}\times A_{m-1}^{(1)})$ in \cite{OS2}.

The aim of this article is to formulate the birational representation of $\widetilde{W}(A_{mn-1}^{(1)}\times A_{m-1}^{(1)}\times A_{m-1}^{(1)})$ again in a framework of a system of $q$-difference equations with $mn\times mn$ matrices.
This formulation is called a Lax form.
A Lax form for $\widetilde{W}(A_{2n+1}^{(1)}\times A_1^{(1)}\times A_1^{(1)})$ was partially given in a framework of the $q$-Drinfeld-Sokolov hierarchy in \cite{Suz2}.
Based on this previous work, we give a complete Lax form for $\widetilde{W}(A_{mn-1}^{(1)}\times A_{m-1}^{(1)}\times A_{m-1}^{(1)})$ in this article.

This article is organized as follows.
In Section \ref{Sec:Bir_Rep} we recall the group of the birational transformations which is isomorphic to $\widetilde{W}(A_{mn-1}^{(1)}\times A_{m-1}^{(1)}\times A_{m-1}^{(1)})$ given in \cite{OS2}.
In Section \ref{Sec:Lax} we introduce a system of $q$-difference equations with $mn\times mn$ matrices, whose compatibility conditions correspond to the group given in Section \ref{Sec:Bir_Rep}.
In Section \ref{Sec:Example} we formulate a generalized $q$-Garnier system in the case $(m,n)=(3,2)$ as an example.

\section{Birational representation}\label{Sec:Bir_Rep}

Recall that the affine Weyl group $W(A^{(1)}_{m-1})$ is generated by the generators $r_i$ $(i\in\mathbb{Z}_m)$ and the fundamental relations
\[
	r_0^2 = r_1^2 = 1,
\]
for $m=2$ and
\[
	r_i^2 = 1,\quad
	r_i\,r_j\,r_i = r_j\,r_i\,r_j\quad (|i-j|=1),\quad
	r_i\,r_j = r_j\,r_i\quad (|i-j|>1),
\]
for $m\geq3$.
Here we denote the quotient ring $\mathbb{Z}/m\mathbb{Z}$ by $\mathbb{Z}_m$.

Let $\varphi_{j,i}$ $(j\in\mathbb{Z}_{mn},i\in\mathbb{Z}_m)$ be dependent variables and $\alpha_j,\beta_i,\beta'_i$ $(j\in\mathbb{Z}_{mn},i\in\mathbb{Z}_m)$ be parameters defined by
\[
	\alpha_j = \prod_{i=0}^{m-1}\varphi_{j,i},\quad
	\beta_i = \prod_{j=0}^{mn-1}\varphi_{j,i},\quad
	\beta'_i = \prod_{j=0}^{mn-1}\varphi_{j,i+j}.
\]
We also set
\[
	\prod_{j=0}^{mn-1}\alpha_j = \prod_{i=0}^{m-1}\beta_i = \prod_{i=0}^{m-1}\beta'_i = \prod_{j=0}^{mn-1}\prod_{i=0}^{m-1}\varphi_{j,i} = q.
\]
Note that the parameters $\alpha_j$, $\beta_i$ and $\beta'_i$ correspond to multiplicative simple roots for $W(A^{(1)}_{mn-1})$, $W(A^{(1)}_{m-1})$ and $W(A^{(1)}_{m-1})$ respectively.
We define birational transformations $r_j$ $(j\in\mathbb{Z}_{mn})$ by
\begin{equation}\begin{split}\label{Eq:act_r}
	&r_j(\varphi_{j-1,i}) = \varphi_{j-1,i}\,\varphi_{j,i+1}\,\frac{P_{j,i+2}}{P_{j,i+1}},\quad
	r_j(\varphi_{j,i}) = \frac{1}{\varphi_{j,i+1}}\frac{P_{j,i}}{P_{j,i+2}},\quad
	r_j(\varphi_{j+1,i}) = \varphi_{j,i}\,\varphi_{j+1,i}\,\frac{P_{j,i+1}}{P_{j,i}}, \\
	&r_j(\varphi_{k,i}) = \varphi_{k,i}\quad (k\neq j,j\pm1),
\end{split}\end{equation}
where
\[
	P_{j,i} = \sum_{k=0}^{m-1}\prod_{l=0}^{k-1}\varphi_{j,i+l}.
\]
They act on the parameters as
\begin{align*}
	&r_j(\alpha_j) = \frac{1}{\alpha_j},\quad
	r_j(\alpha_{j\pm1}) = \alpha_{j\pm1}\,\alpha_j,\quad
	r_j(\alpha_k) = \alpha_k\quad (k\neq j,j\pm1),\quad
	r_j(\beta_i) = \beta_i,\quad
	r_j(\beta'_i) = \beta'_i.
\end{align*}
We also define birational transformations $s_i,s'_i$ $(i\in\mathbb{Z}_m)$ by
\begin{equation}\begin{split}\label{Eq:act_s_2}
	&s_i(\varphi_{j,i}) = \frac{1}{\varphi_{j+1,i}}\frac{Q_{j,i}}{Q_{j+2,i}},\quad
	s_i(\varphi_{j,i+1}) = \varphi_{j,i}\,\varphi_{j,i+1}\,\varphi_{j+1,i}\,\frac{Q_{j+2,i}}{Q_{j,i}}, \\
	&s'_i(\varphi'_{j,i}) = \frac{1}{\varphi'_{j+1,i}}\frac{Q'_{j,i}}{Q'_{j+2,i}},\quad
	s'_i(\varphi'_{j,i+1}) = \varphi'_{j,i}\,\varphi'_{j,i+1}\,\varphi'_{j+1,i}\,\frac{Q'_{j+2,i}}{Q'_{j,i}},
\end{split}\end{equation}
for $m=2$ and
\begin{equation}\begin{split}\label{Eq:act_s_3}
	&s_i(\varphi_{j,i-1}) = \varphi_{j,i-1}\,\varphi_{j+1,i}\,\frac{Q_{j+2,i}}{Q_{j+1,i}},\quad
	s_i(\varphi_{j,i}) = \frac{1}{\varphi_{j+1,i}}\frac{Q_{j,i}}{Q_{j+2,i}},\quad
	s_i(\varphi_{j,i+1}) = \varphi_{j,i}\,\varphi_{j,i+1}\,\frac{Q_{j+1,i}}{Q_{j,i}}, \\
	&s_i(\varphi_{j,k}) = \varphi_{j,k}\quad (k\neq i,i\pm1), \\
	&s'_i(\varphi'_{j,i-1}) = \varphi'_{j,i-1}\,\varphi'_{j+1,i}\,\frac{Q'_{j+2,i}}{Q'_{j+1,i}},\quad
	s'_i(\varphi'_{j,i}) = \frac{1}{\varphi'_{j+1,i}}\frac{Q'_{j,i}}{Q'_{j+2,i}},\quad
	s'_i(\varphi'_{j,i+1}) = \varphi'_{j,i}\,\varphi'_{j,i+1}\,\frac{Q'_{j+1,i}}{Q'_{j,i}}, \\
	&s'_i(\varphi'_{j,k}) = \varphi'_{j,k}\quad (k\neq i,i\pm1),
\end{split}\end{equation}
for $m\geq3$, where
\[
	Q_{j,i} = \sum_{k=0}^{mn-1}\prod_{l=0}^{k-1}\varphi_{j+l,i},\quad
	Q'_{j,i} = \sum_{k=0}^{mn-1}\prod_{l=0}^{k-1}\varphi'_{j+l,i}.
\]
and $\varphi'_{j,i}=\varphi_{-j,i-j}$.
They act on the parameters as
\begin{align*}
	&s_i(\beta_i) = \frac{1}{\beta_i},\quad
	s_i(\beta_{i+1}) = \beta_{i+1}\,\beta_i^2,\quad
	s_i(\alpha_j) = \alpha_j,\quad
	s_i(\beta'_k) = \beta'_k, \\
	&s'_i(\beta'_i) = \frac{1}{\beta'_i},\quad
	s'_i(\beta'_{i+1}) = \beta'_{i+1}(\beta'_i)^2,\quad
	s'_i(\alpha_j) = \alpha_j,\quad
	s'_i(\beta_k) = \beta_k,
\end{align*}
for $m=2$ and
\begin{align*}
	&s_i(\beta_i) = \frac{1}{\beta_i},\quad
	s_i(\beta_{i\pm1}) = \beta_{i\pm1}\,\beta_i,\quad
	s_i(\beta_k) = \beta_k\quad (k\neq i,i\pm1),\quad
	s_i(\alpha_j) = \alpha_j,\quad
	s_i(\beta'_l) = \beta'_l, \\
	&s'_i(\beta'_i) = \frac{1}{\beta'_i},\quad
	s'_i(\beta'_{i\pm1}) = \beta'_{i\pm1}\,\beta'_i,\quad
	s'_i(\beta'_k) = \beta'_k\quad (k\neq i,i\pm1),\quad
	s'_i(\alpha_j) = \alpha_j,\quad
	s'_i(\beta_l) = \beta_l,
\end{align*}
for $m\geq3$.

\begin{fact}[{\cite{IIO,MOT}}]
If we set
\[
	G = \langle r_0,\ldots,r_{mn-1}\rangle,\quad
	H = \langle s_0,\ldots,s_{m-1}\rangle,\quad
	H' = \langle s'_0,\ldots,s'_{m-1}\rangle,
\]
then the groups $G,$ $H$ and $H'$ are isomorphic to the affine Weyl groups $W(A^{(1)}_{mn-1})$, $W(A^{(1)}_{m-1})$ and $W(A^{(1)}_{m-1})$ respectively.
Moreover, any two groups are mutually commutative.
\end{fact}

\begin{rem}
In this article we interpret a composition of transformations in terms of automorphisms of the field of rational functions $\mathbb{C}(\varphi_{j,i})$.
For example, the compositions $r_0\,r_1,r_1\,r_0$ act on the parameter $\alpha_0$ as
\[
	r_0\,r_1(\alpha_0) = r_0(\alpha_0\,\alpha_1) = r_0(\alpha_0)\,r_0(\alpha_1) = \alpha_1,\quad
	r_1\,r_0(\alpha_0) = r_1\left(\frac{1}{\alpha_0}\right) = \frac{1}{r_1(\alpha_0)} = \frac{1}{\alpha_0\,\alpha_1}.
\]
\end{rem}

In addition, we define birational transformations $\pi_1,\pi_2$ by
\begin{align}
	\pi_1(\varphi_{j,i}) &= \varphi_{j+1,i+1}, \label{Eq:act_pi1} \\
	\pi_2(\varphi_{j,i}) &= \varphi_{j,i+1}. \label{Eq:act_pi2}
\end{align}
They act on the parameters as
\begin{align*}
	&\pi_1(\alpha_j) = \alpha_{j+1},\quad
	\pi_1(\beta_i) = \beta_{i+1},\quad
	\pi_1(\beta'_i) = \beta'_i, \\
	&\pi_2(\alpha_j) = \alpha_j,\quad
	\pi_2(\beta_i) = \beta_{i+1},\quad
	\pi_2(\beta'_i) = \beta'_{i+1}.
\end{align*}

\begin{fact}[{\cite{OS2}}]
The transformations $\pi_1,\pi_2$ satisfy fundamental relations
\begin{align*}
	&\pi_1^{mn} = 1,\quad
	\pi_2^m = 1,\quad
	\pi_1\,\pi_2 = \pi_2\,\pi_1, \\
	&r_j\,\pi_1 = \pi_1\,r_{j-1},\quad
	s_i\,\pi_1 = \pi_1\,s_{i-1},\quad
	s'_i\,\pi_1 = \pi_1\,s'_i, \\
	&r_j\,\pi_2 = \pi_2\,r_j,\quad
	s_i\,\pi_2 = \pi_2\,s_{i-1},\quad
	s'_i\,\pi_2 = \pi_2\,s'_{i-1},
\end{align*}
for $j\in\mathbb{Z}_{mn}$ and $i\in\mathbb{Z}_m$.
\end{fact}

Hence we can regard the semi-direct product $\langle G,H,H'\rangle\rtimes\langle\pi_1,\pi_2\rangle$ as an extended affine Weyl group $\widetilde{W}(A_{mn-1}^{(1)}\times A_{m-1}^{(1)}\times A_{m-1}^{(1)})$.
On the other hand, the group $\widetilde{W}(A_{mn-1}^{(1)}\times A_{m-1}^{(1)}\times A_{m-1}^{(1)})$ contains an abelian normal subgroup generated by translations.
Combining them, we can provide a class of generalized $q$-Garnier systems as translations.
Recall that the $q$-Garnier system was obtained as the case $m=2$ in \cite{OS1}.

\section{Lax form}\label{Sec:Lax}

Let us introduce an independent variable $z$ satisfying
\[
	r_j(z) = z\quad (j\in\mathbb{Z}_{mn}),\quad
	s_i(z) = z,\quad
	s'_i(z) = z\quad (i\in\mathbb{Z}_m),\quad
	\pi_1(z) = z,
\]
and
\begin{equation}\label{Eq:act_pi2_z}
	\pi_2(z) = q^{1/m}z.
\end{equation}
We also set
\[
	\zeta = z\prod_{j=1}^{mn-1}\prod_{i=0}^{m-2}\alpha_j^{(mn-j)/m}\beta_i^{(i+1)/m}.
\]
Then the birational transformations act on $\zeta$ as
\begin{align*}
	&r_0(\zeta) = \alpha_0^n\,\zeta,\quad
	r_1(\zeta) = \frac{1}{\alpha_1^n}\,\zeta,\quad
	r_k(\zeta) = \zeta\quad (k\neq0,1), \\
	&s_{m-2}(\zeta) = \frac{1}{\beta_{m-2}}\,\zeta,\quad
	s_{m-1}(\zeta) = \beta_{m-1}\,\zeta,\quad
	s_k(\zeta) = \zeta\quad (k\neq m-2,m-1),\quad
	s'_i(\zeta) = \zeta, \\
	&\pi_1(\zeta) = \frac{\beta_{m-1}}{\alpha_1^n}\,\zeta,\quad
	\pi_2(\zeta) = \beta_{m-1}\,\zeta.
\end{align*}

We first give a Lax form for the birational transformations $\pi_1,\pi_2$.
Let $E_{j_1,j_2}$ be a $mn\times mn$ matrix with $1$ in $(j_1,j_2)$-th entry and $0$ elsewhere.
Consider $mn\times mn$ matrices
\[
	\Pi_1 = \zeta^{-\log_q\alpha_1}\left(\sum_{j=1}^{mn-1}\prod_{k=0}^{j-1}\frac{1}{\varphi_{1,k}}\,E_{j,j+1}+\frac{1}{\alpha_1^n\prod_{i=0}^{m-2}\beta_i}\,\zeta\,E_{mn,1}\right),
\]
and
\[
	\Pi_2 = \sum_{j=1}^{mn}\prod_{k=1}^{j-1}\varphi_{k,j-1}\,E_{j,j} + \sum_{j=1}^{mn-1}E_{j,j+1} + \frac{1}{\prod_{i=0}^{m-2}\beta_i}\,\zeta\,E_{mn,1}.
\]
We also set
\[
	M = \pi_2^{m-1}(\Pi_2)\,\pi_2^{m-2}(\Pi_2)\,\ldots\,\pi_2(\Pi_2)\,\Pi_2,
\]
and $T_{q,z}=\pi_2^m$.

\begin{thm}
Under transformations \eqref{Eq:act_pi2} and \eqref{Eq:act_pi2_z}, the compatibility condition of a system of linear $q$-difference equations
\[
	T_{q,z}(\psi) = M\,\psi,\quad
	\pi_1(\psi) = \Pi_1\,\psi,
\]
with a fundamental relation
\begin{equation}\label{Eq:comp_pi1_pi2}
	\pi_2(\Pi_1)\,\Pi_2 = \beta_{m-1}^{-\log_q\alpha_1}\varphi_{1,0}\,\pi_1(\Pi_2)\,\Pi_1,
\end{equation}
is equivalent to transformation \eqref{Eq:act_pi1}.
\end{thm}

\begin{proof}
If we assume transformation \eqref{Eq:act_pi1}, then we can show system \eqref{Eq:comp_pi1_pi2} by a direct calculation.
It follows that
\begin{align*}
	T_{q,z}(\Pi_1) &= \pi_2^{m-1}\pi_2(\Pi_1) \\
	&= \beta_{m-2}^{-\log_q\alpha_1}\varphi_{1,m-1}\,\pi_2^{m-1}\pi_1(\Pi_2)\,\pi_2^{m-1}(\Pi_1)\,\pi_2^{m-1}(\Pi_2^{-1}) \\
	&= \beta_{m-3}^{-\log_q\alpha_1}\beta_{m-2}^{-\log_q\alpha_1}\varphi_{1,m-2}\,\varphi_{1,m-1}\,\pi_2^{m-1}\pi_1(\Pi_2)\,\pi_2^{m-2}\pi_1(\Pi_2)\,\pi_2^{m-2}(\Pi_1)\,\pi_2^{m-2}(\Pi_2^{-1})\,\pi_2^{m-1}(\Pi_2^{-1}) \\
	&= \ldots \\
	&= q^{-\log_q\alpha_1}\alpha_1\,\pi_2^{m-1}\pi_1(\Pi_2)\,\ldots\,\pi_2\pi_1(\Pi_2)\,\pi_1(\Pi_2)\,\Pi_1\,\Pi_2^{-1}\,\pi_2(\Pi_2^{-1})\,\ldots\,\pi_2^{m-1}(\Pi_2^{-1}) \\
	&= \pi_1(M)\,\Pi_1\,M^{-1}.
\end{align*}
Inversely, we can restore transformation \eqref{Eq:act_pi1} by going back the way we came.
\end{proof}

We next give a Lax form for the other birational transformations.
Let $I$ be the identity matrix.
Consider $mn\times mn$ matrices
\begin{align*}
	R_0 &= \zeta^{\log_q\alpha_0}\left(\sum_{j=1}^{mn}\frac{P_{0,j-1}\prod_{k=0}^{j-2}\varphi_{0,k}}{P_{0,0}}E_{j,j}+\frac{q\,(1-\alpha_0)}{\varphi_{0,m-1}\,P_{0,0}}\,\frac{1}{\zeta}\,E_{1,mn}\right), \\
	R_1 &= \zeta^{-\log_q\alpha_1}\left(\sum_{j=1}^{mn}\frac{\varphi_{1,1}\,P_{1,2}}{P_{1,j}\prod_{k=1}^{j-1}\varphi_{1,k}}E_{j,j}+\frac{1-\alpha_1}{P_{1,1}}\,E_{2,1}\right), \\
	R_j &= I + \frac{(1-\alpha_j)\prod_{k=1}^{j-1}\varphi_{k,j-1}}{P_{j,j}}\,E_{j+1,j}\quad (j=2,\ldots,mn-1),
\end{align*}
and
\begin{align*}
	S_i &= I + \sum_{k=0}^{n-1}\left(\frac{Q_{mk+i+2,i}\prod_{l=2}^{mk+i+1}\varphi_{l,i}}{Q_{2,i}}-1\right)E_{mk+i+1,mk+i+1} \\
	&\quad + \sum_{k=0}^{n-1}\left(\frac{Q_{1,i}}{Q_{mk+i+2,i}\prod_{l=1}^{mk+i+1}\varphi_{l,i}}-1\right)E_{mk+i+2,mk+i+2} \\
	&\quad + \sum_{k=0}^{n-1}\frac{\beta_i-1}{\varphi_{1,i}\,Q_{2,i}}\,E_{mk+i+1,mk+i+2}\quad (i=0,\ldots,m-2), \\
	S_{m-1} &= I + \sum_{k=1}^{n}\left(\frac{Q_{mk+1,m-1}\prod_{l=2}^{mk}\varphi_{l,m-1}}{Q_{2,m-1}}-1\right)E_{mk,mk} \\
	&\quad + \sum_{k=1}^{n-1}\left(\frac{Q_{1,m-1}}{Q_{mk+1,m-1}\prod_{l=1}^{mk}\varphi_{l,m-1}}-1\right)E_{mk+1,mk+1} \\
	&\quad + \sum_{k=1}^{n-1}\frac{\beta_{m-1}-1}{\varphi_{1,m-1}\,Q_{2,m-1}}\,E_{mk,mk+1} + \frac{\beta_{m-1}-1}{\varphi_{1,m-1}\,Q_{2,m-1}\prod_{l=0}^{m-2}\beta_l}\,\zeta\,E_{mn,1}, \\
	S'_i &= I + \sum_{k=0}^{n-1}\left(\frac{\varphi'_{0,i}\,Q'_{1,i}}{Q'_{0,i}}-1\right)E_{mk+i+2,mk+i+2}\quad (i=0,\ldots,m-2).
\end{align*}
We also set
\[
	S'_{m-1} = s'_{m-1}(\Pi_2^{-1})\,\pi_2(S'_{m-2})\,\Pi_2.
\]
The explicit formula of $S'_{m-1}$ is not given here.

\begin{rem}
Strangely the matrix $S'_{m-1}$ is rational in $\zeta$, is not diagonal and hence is much more complicated than the others.
The cause has not been clarified yet.
However, this is not a serious matter.
We can avoid using the matrix $S'_{m-1}$, or equivalently the transformation $s'_{m-1}$, when we define the group of translations.
\end{rem}

\begin{thm}
Under transformations \eqref{Eq:act_pi2} and \eqref{Eq:act_pi2_z}, the compatibility condition of a system of linear $q$-difference equations
\[
	T_{q,z}(\psi) = M\,\psi,\quad
	r_j(\psi) = R_j\,\psi\quad (j\in\mathbb{Z}_{mn}),\quad
	s_i(\psi) = S_i\,\psi,\quad
	s'_i(\psi) = S'_i\,\psi\quad (i\in\mathbb{Z}_m),
\]
with fundamental relations
\begin{equation}\begin{split}\label{Eq:comp_rs_pi2}
	&\pi_2(R_j)\,\Pi_2 = r_j(\Pi_2)\,R_j\quad (j\in\mathbb{Z}_{mn}), \\
	&\pi_2(S_{i-1})\,\Pi_2 = s_i(\Pi_2)\,S_i,\quad
	\pi_2(S'_{i-1})\,\Pi_2 = s'_i(\Pi_2)\,S'_i\quad (i\in\mathbb{Z}_m),
\end{split}\end{equation}
is equivalent to transformations \eqref{Eq:act_r}, \eqref{Eq:act_s_2} and \eqref{Eq:act_s_3}.
\end{thm}

\begin{proof}
We prove the formulae for the transformations $s_i$.
If we assume transformations \eqref{Eq:act_s_2} and \eqref{Eq:act_s_3}, then we can show the second equation of \eqref{Eq:comp_rs_pi2} by a direct calculation with
\[
	\varphi_{j,i}\,Q_{j+1,i} = Q_{j,i} + \beta_j + 1.
\]
It follows that
\begin{align*}
	S_i &= s_i(\Pi_2^{-1})\,\pi_2(S_{i-1})\,\Pi_2 \\
	&= s_i(\Pi_2^{-1})\,\pi_2s_{i-1}(\Pi_2^{-1})\,\pi_2^2(S_{i-2})\,\pi_2(\Pi_2)\,\Pi_2 \\
	&= \ldots \\
	&= s_i(\Pi_2^{-1})\,\pi_2s_{i-1}(\Pi_2^{-1})\,\ldots\,\pi_2^{m-1}s_{i+1}(\Pi_2^{-1})\,\pi_2^m(S_i)\,\pi_2^{m-1}(\Pi_2)\,\ldots\,\pi_2(\Pi_2)\,\Pi_2 \\
	&= s_i(M^{-1})\,T_{q,z}(S_i)\,M.
\end{align*}
Recall that $\pi_2\,s_{i-1}=s_i\,\pi_2$.
Inversely, we can restore transformations \eqref{Eq:act_s_2} and \eqref{Eq:act_s_3} by going back the way we came.

We can prove the formulae for the other transformations in a similar manner by using
\[
	\varphi_{j,i}\,P_{j,i+1} = P_{j,i} + \alpha_j + 1\quad (j\in\mathbb{Z}_{mn},\ i\in\mathbb{Z}_m),
\]
and
\[
	Q'_{j,i} + \varphi'_{j+1,i}\,Q'_{j+2,i} = (1+\varphi'_{j,i})\,Q'_{j+1,i}\quad (j\in\mathbb{Z}_{mn},\ i\in\mathbb{Z}_m).
\]
We don't state its detail here.
\end{proof}

\begin{rem}
The above definition of the matrices is suggested by \cite{Suz2}, in which we propose a $q$-analogue of the Drinfeld-Sokolov hierarchy of type $A_{mn-1}^{(1)}$ corresponding to the partition $(n,\ldots,n)$ of $mn$.
A relationship between our Lax form and the cluster algebra is not completely understood.
It is a future problem.
\end{rem}

\begin{rem}
The transformations $r_j$ were given by gauge transformations of the $q$-KP hierarchy in \cite{KNY1}.
Besides, the transformations $s_i$ (or $s'_i$) were given with the aid of a factorization of a matrix and permutations of the factors in \cite{KNY1,KNY2,P1,P2}.
\end{rem}

\section{Example}\label{Sec:Example}

In the previous section we obtained a Lax form for the extended affine Weyl group $\langle G,H,H'\rangle\rtimes\langle\pi_1,\pi_2\rangle$.
Hence we can derive a class of generalized $q$-Garnier systems together with their Lax pairs from the group of translations systematically.
Since the case $m=2$ has been already studied in detail in \cite{OS1,OS2}, we consider another case as an example in this section.

\begin{rem}
In Remark 4.2 of \cite{OS1}, we conjectured that the translation $s_1\,r_1\,\ldots\,r_{2n-1}\,\pi_1$ for $m=2$ implies a variation of the $q$-Garnier system given in \S3.2.4 of \cite{NY2}.
We can show that this conjecture is true by using the Lax form.
Since the proof can be given in a similar manner as that given in \S6 of \cite{OS1}, we omit it here.
\end{rem}

Let $m=3$ and $n=2$.
Then the matrix $M$ is described as
\[
	M = \begin{pmatrix}1&P_{1,1}&P^{*}_{1,2}&1&0&0\\0&\alpha_1&\frac{\alpha_1\,P_{2,2}}{\varphi_{1,1}}&\varphi_{1,0}\,P^{*}_{2,0}&1&0\\0&0&\alpha_1\,\alpha_2&\frac{\alpha_1\,\alpha_2\,P_{3,0}}{\varphi_{1,2}\,\varphi_{2,2}}&\varphi_{1,1}\,\varphi_{2,1}\,P^{*}_{3,1}&1\\\frac{\zeta}{\beta_0\,\beta_1}&0&0&\alpha_1\,\alpha_2\,\alpha_3&\frac{\alpha_1\,\alpha_2\,\alpha_3\,P_{4,1}}{\varphi_{1,0}\,\varphi_{2,0}\,\varphi_{3,0}}&\varphi_{1,2}\,\varphi_{2,2}\,\varphi_{3,2}\,P^{*}_{4,2}\\\frac{\varphi_{1,0}\,\varphi_{2,0}\,\varphi_{3,0}\,\varphi_{4,0}\,P^{*}_{5,0}\,\zeta}{\beta_0\,\beta_1}&\frac{\zeta}{\beta_1}&0&0&\alpha_1\,\alpha_2\,\alpha_3\,\alpha_4&\frac{\alpha_1\,\alpha_2\,\alpha_3\,\alpha_4\,P_{5,2}}{\varphi_{1,1}\,\varphi_{2,1}\,\varphi_{3,1}\,\varphi_{4,1}}\\\frac{\alpha_1\,\alpha_2\,\alpha_3\,\alpha_4\,\alpha_5\,P_{0,0}\,\zeta}{\beta_0\,\beta_1\,\varphi_{1,2}\,\varphi_{2,2}\,\varphi_{3,2}\,\varphi_{4,2}\,\varphi_{5,2}}&\frac{\varphi_{1,1}\,\varphi_{2,1}\,\varphi_{3,1}\,\varphi_{4,1}\,\varphi_{5,1}\,P^{*}_{0,1}\,\zeta}{\beta_1}&\zeta&0&0&\alpha_1\,\alpha_2\,\alpha_3\,\alpha_4\,\alpha_5\end{pmatrix},
\]
where
\[
	P_{j,i} = 1 + \varphi_{j,i} + \varphi_{j,i}\,\varphi_{j,i+1},\quad
	P^{*}_{j,i} = 1 + \varphi_{j,i} + \varphi_{j,i}\,\varphi_{j+1,i}.
\]
Note that the dependent variables and the parameters satisfy a periodic condition
\[
	\varphi_{j,i} = \varphi_{j+6,i} = \varphi_{j,i+3},\quad
	\alpha_j = \alpha_{j+6},\quad
	\beta_i = \beta_{i+3},\quad
	\beta'_i = \beta'_{i+3}.
\]
We consider a translation
\[
	\tau = s_0\,s_1\,s'_0\,s'_1\,\pi_2,
\]
which acts on the parameters as
\[
	\tau(\alpha_j) = \alpha_j,\quad
	\tau(\beta_i) = q^{-\delta_{i,0}+\delta_{i,2}}\,\beta_i,\quad
	\tau(\beta'_i) = q^{-\delta_{i,0}+\delta_{i,2}}\,\beta'_i,
\]
where $\delta_{i,k}$ stands for the Kronecker's delta.
Then the action of $\tau$ on the dependent variables $\varphi_{j,i}$ is derived from the compatibility condition of a Lax pair
\begin{equation}\label{Eq:EO4_Painleve}
	T_{q,z}(T)\,M=\tau(M)\,T,\quad
	T = s_1s'_0s'_1\pi_2(S_0)\,s'_0s'_1\pi_2(S_1)\,s'_1\pi_2(S'_0)\,\pi_2(S'_1)\,\Pi_2.
\end{equation}
System \eqref{Eq:EO4_Painleve} can be regarded as a system of meromorphic $q$-difference equations of eighth order; see Remark 4.3 below.
We don't give its explicit formula here.

In the following, we introduce a particular solution of system \eqref{Eq:EO4_Painleve} in terms of a linear $q$-difference equation.
Assume that
\[
	P_{1,1} = P^{*}_{1,2} = P_{4,1} = P^{*}_{4,2} = 0,
\]
which contains a constraint between parameters
\[
	\alpha_1^2\,\alpha_2\,\alpha_4^2\,\alpha_5 = \frac{\beta_0\,\beta'_2}{\beta_2\,\beta'_0}.
\]
Also assume that
\[
	\beta'_0\,\beta'_2 = 1.
\]
Then there exist two invariants
\[
	\tau(\varphi_{0,0}\,\varphi_{0,2}\,\varphi_{1,0}) = \varphi_{0,0}\,\varphi_{0,2}\,\varphi_{1,0},\quad
	\tau(\varphi_{1,2}\,\varphi_{2,0}\,\varphi_{2,2}) = \varphi_{1,2}\,\varphi_{2,0}\,\varphi_{2,2}.
\]
Now we normalize that
\[
	\varphi_{0,0}\,\varphi_{0,2}\,\varphi_{1,0} = c,\quad
	\varphi_{1,2}\,\varphi_{2,0}\,\varphi_{2,2} = -\frac{(c-1)\,\alpha_2\,\beta_0\,\beta'_2}{c\,\alpha_1\,\alpha_2\,\alpha_4\,\alpha_5-\beta_0\,\beta'_2},
\]
where $c$ is an arbitrary constant.
Then we obtain a system of meromorphic $q$-difference equations with one independent variable $\beta'_2$, three dependent variables $(1+\varphi_{0,0})\,\varphi_{0,2},\frac{\varphi_{2,0}}{\varphi_{3,2}},(1+\varphi_{3,0})\,\varphi_{2,0}$ and six parameters $\alpha_1,\alpha_2,\alpha_3,\alpha_4,\alpha_5,\beta_0\,\beta'_2$.
We don't give its explicit formula here.

Let $x_0,x_1,x_2,x_3$ be dependent variables such that
\begin{align*}
	\frac{\tau(x_0)}{x_0} &= \frac{\beta'_2-\alpha_1}{\alpha_1\,\beta'_2} - \frac{-c\,\alpha_1\,\alpha_2\,\alpha_4\,\alpha_5\,(\alpha_4-1)+\alpha_1\,\alpha_2\,\alpha_4^2\,\alpha_5-\beta_0\,\beta'_2}{(c-1)\,\alpha_1^2\,\alpha_2^2\,\alpha_4^2\,\alpha_5}\,\frac{\varphi_{2,0}}{\varphi_{3,2}} \\
	&\quad + \frac{c\,(\alpha_1\,\alpha_2\,\alpha_4\,\alpha_5-\beta_0\,\beta'_2)}{(c-1)\,\alpha_1\,\alpha_2\,\beta_0\,\beta'_2}\,(1+\varphi_{3,0})\,\varphi_{2,0},
\end{align*}
and
\begin{align*}
	\frac{x_1}{x_0} &= \frac{1}{c-1}\,(1+\varphi_{0,0})\,\varphi_{0,2} - \frac{c\,(\alpha_0-1)+\alpha_0\,(\alpha_1-1)}{(c-1)(\alpha_0\,\alpha_1-1)}, \\
	\frac{x_2}{x_0} &= \frac{c\,\alpha_1\,\alpha_2\,\alpha_4\,\alpha_5-\beta_0\,\beta'_2}{c-1}\,\frac{\varphi_{2,0}}{\varphi_{3,2}}, \\
	\frac{x_3}{x_0} &= \frac{c\,\alpha_1\,\alpha_2\,\alpha_3\,\alpha_4\,\alpha_5\,(\alpha_4-1)+\beta_0\,\beta'_2\,(\alpha_3-1)}{(c-1)\,\alpha_1\,\alpha_2\,\alpha_4\,\alpha_5\,(\alpha_3\,\alpha_4-1)}\,\frac{\varphi_{2,0}}{\varphi_{3,2}} + \frac{c}{c-1}(1+\varphi_{3,0})\,\varphi_{2,0}.
\end{align*}
Recall that $\alpha_0\,\alpha_1\,\alpha_2\,\alpha_3\,\alpha_4\,\alpha_5=q$.

\begin{prop}
A vector of variables $x={}^t(x_0,x_1,x_2,x_3)$ satisfies a system of linear $q$-difference equation
\begin{equation}\label{Eq:EO4}
	\tau(x) = \left(A_0+\beta'_2\,A_1\right)x,
\end{equation}
with $4\times4$ matrices
\begin{align*}
	A_0 &= \begin{pmatrix}\frac{1}{\alpha_1}&0&-\frac{(\alpha_4-1)(\alpha_1\,\alpha_2\,\alpha_3\,\alpha_4^2\,\alpha_5-\beta_0\,\beta'_2)}{\alpha_1^2\,\alpha_2^2\,\alpha_4^2\,\alpha_5\,\beta_0\,\beta'_2\,(\alpha_3\,\alpha_4-1)}&\frac{\alpha_1\,\alpha_2\,\alpha_4\,\alpha_5-\beta_0\,\beta'_2}{\alpha_1\,\alpha_2\,\beta_0\,\beta'_2}\\0&\alpha_0&-\frac{(\alpha_0-1)(\alpha_4-1)(\alpha_0\,\alpha_1^2\,\alpha_2\,\alpha_3\,\alpha_4^2\,\alpha_5-\beta_0\,\beta'_2)}{\alpha_1^2\,\alpha_2^2\,\alpha_4^2\,\alpha_5\,\beta_0\,\beta'_2\,(\alpha_0\,\alpha_1-1)(\alpha_3\,\alpha_4-1)}&\frac{(\alpha_0-1)(\alpha_0\,\alpha_1^2\,\alpha_2\,\alpha_4\,\alpha_5-\beta_0\,\beta'_2)}{\alpha_1\,\alpha_2\,\beta_0\,\beta'_2\,(\alpha_0\,\alpha_1-1)}\\0&0&\frac{1}{\alpha_1\,\alpha_2\,\alpha_3\,\alpha_4}&0\\0&0&0&\frac{1}{\alpha_1\,\alpha_2}\end{pmatrix}, \\
	A_1 &= \begin{pmatrix}-1&0&0&0\\0&-1&0&0\\-\frac{(\alpha_1-1)(\alpha_0\,\alpha_1^2\,\alpha_2\,\alpha_4\,\alpha_5-\beta_0\,\beta'_2)}{\alpha_1\,\alpha_3\,(\alpha_0\,\alpha_1-1)}&-\frac{\alpha_1\,\alpha_2\,\alpha_4\,\alpha_5-\beta_0\,\beta'_2}{\alpha_3}&-\frac{1}{\alpha_1\,\alpha_2\,\alpha_3}&0\\\frac{(\alpha_1-1)(\alpha_3-1)(\alpha_0\,\alpha_1^2\,\alpha_2\,\alpha_3\,\alpha_4^2\,\alpha_5-\beta_0\,\beta'_2)}{\alpha_1^2\,\alpha_2\,\alpha_3\,\alpha_4\,\alpha_5\,(\alpha_0\,\alpha_1-1)(\alpha_3\,\alpha_4-1)}&\frac{(\alpha_3-1)(\alpha_1\,\alpha_2\,\alpha_3\,\alpha_4^2\,\alpha_5-\beta_0\,\beta'_2)}{\alpha_1\,\alpha_2\,\alpha_3\,\alpha_4\,\alpha_5\,(\alpha_3\,\alpha_4-1)}&0&-\frac{1}{\alpha_1\,\alpha_2\,\alpha_3}\end{pmatrix}.
\end{align*}
\end{prop}

We can prove this proposition by a direct calculation.

\begin{rem}
In a continuous limit $q\to1$, system \eqref{Eq:EO4} reduces to the rigid system with the spectral type $\{22,211,1111\}$ whose solution is expressed in terms of Simpson's even four hypergeometric function.
Moreover, system \eqref{Eq:EO4_Painleve} reduces to the isomonodromy deformation equation of eighth order with the spectral type $\{31,31,31,22,1111\}$.
\end{rem}

\begin{rem}
Another generalization of the $q$-Garnier system is proposed together with its particular solution expressed in terms of the $q$-hypergeometric function $\mathcal{F}_{N,M}$ in \cite{P1,P2}.
We have not clarified a relationship between our result and that of \cite{P1,P2} yet.
It is a future problem.
\end{rem}

\section*{Acknowledgement}

This work was supported by JSPS KAKENHI Grant Number 20K03645.


\end{document}